\documentclass[letterpaper, 10 pt, conference]{ieeeconf}

\IEEEoverridecommandlockouts

\usepackage{cite}
\usepackage{amsmath,amssymb,amsfonts}

\usepackage{amsthm}
\usepackage{float}
\newcommand{\field}[1]{\mathbb{#1}}
\newcommand{\R}{\field{R}}

\usepackage{graphicx}
\usepackage{microtype}

\usepackage{xcolor}
\usepackage{array}
\usepackage{booktabs}
\usepackage{hyperref}
\usepackage[nameinlink,capitalize,noabbrev]{cleveref}

\makeatletter
\g@addto@macro\normalsize{\abovedisplayskip=2pt plus 1pt minus 1pt \belowdisplayskip=2pt plus 1pt minus 1pt \abovedisplayshortskip=1pt \belowdisplayshortskip=1pt}
\g@addto@macro\small{\abovedisplayskip=2pt plus 1pt minus 1pt \belowdisplayskip=2pt plus 1pt minus 1pt \abovedisplayshortskip=1pt \belowdisplayshortskip=1pt}
\makeatother
\DeclareMathOperator*{\argmin}{arg\,min}

\author{Turki Bin Mohaya$^\star$ \quad Maitham F. AL-Sunni$^\star$ \quad John M. Dolan \quad Peter Seiler%
\thanks{$^\star$ These authors contributed equally to the work.}%
\thanks{T. Bin Mohaya and P. Seiler are with the Department of Electrical Engineering and Computer Science at the University of Michigan, Ann Arbor, MI, USA. Email: {\tt \{turki,pseiler\}@umich.edu}}%
\thanks{M. F. AL-Sunni is with the Department of Electrical \& Computer Engineering, Carnegie Mellon University, Pittsburgh, PA,
USA. Email: {\tt malsunni@andrew.cmu.edu}}%
\thanks{J. M. Dolan is with the Robotics Institute, Carnegie Mellon University, Pittsburgh, PA, USA. Email: {\tt jdolan@andrew.cmu.edu}}%
}

\newtheorem{assumption}{Assumption}
\newtheorem{theorem}{Theorem}
\newtheorem{lemma}{Lemma}
\newtheorem{corollary}{Corollary}

\begin{document}
\title{\LARGE \bf Generalizable Optimal Control with Transformers: \\ Closed-Loop Certification and Near-Optimality Guarantees}

\maketitle

\begin{abstract}
\looseness=-1 This letter develops closed-loop performance certificates for a transformer-based feedback policy. The policy is trained to imitate optimal Linear Quadratic Regulator (LQR) control across a family of heterogeneous Multiple-Input, Multiple-Output (MIMO) Linear Time-Invariant (LTI) systems. First, we establish a finite-sample excess-risk bound for the imitation loss minimized during training. Second, for each fixed problem instance, we derive regional closed-loop guarantees consisting of a forward-invariant operating region and a worst-case bound on deviation from the optimal rollout. Our main result is a probabilistic certificate for finite-horizon closed-loop near-optimality. Using an exact LQR cost identity, we express excess cost as a measurable per-rollout statistic and use independent calibration and validation rollouts to obtain a high-confidence bound on its violation probability. We evaluate the certificate on $28$ benchmark systems. This uses the base policy on seen systems and system-specific fine-tuned copies on unseen systems, with each rollout drawing the plant, cost, and initial condition from the corresponding certification distribution. All per-system certificates have violation probabilities below $3.1\%$, each at $95\%$ confidence; twenty systems certify suboptimality below $10\%$, with the tightest threshold equal to $4.8\times10^{-6}$.
\end{abstract}

\section{Introduction}
\label{sec:introduction}

\looseness=-1 Learned controllers are increasingly trained to imitate optimal control laws, yet they ship with training-loss numbers rather than closed-loop certificates. This gap has three sources. The imitation loss is measured on expert-generated trajectories, whereas the deployed policy generates its own state distribution. Moreover, a bounded training loss cannot directly upper-bound an unbounded quadratic cost. Finally, standard generalization bounds are relative to an unknown best-in-class risk and therefore do not provide an absolute closed-loop performance guarantee. For the Linear Quadratic Regulator (LQR), however, the optimal feedback and its cost are known exactly. This letter shows that the gap can be closed. A validation-measured excess statistic, grounded in an exact cost identity, delivers an explicit high-confidence bound on closed-loop LQR sub-optimality over a validation horizon. Regional guarantees of a certified operating region complement it.

The policy we certify is a single transformer trained to imitate optimal LQR state feedback across a family of heterogeneous MIMO LTI systems of differing state and input dimensions and cost objectives. It maps a window of recent states and the cost matrices $(Q,R)$ to a control action, using no plant matrices at inference. In our initial work~\cite{binmohaya2026transformers}, the policy is stabilizing in all $9{,}675$ evaluation rollouts across $28$ systems, with median relative sub-optimality $0.022\%$ on seen systems. The results in~\cite{binmohaya2026transformers} were empirical studies via simulations. Here we provide theory to support those empirical observations. The contributions below detail three specific performance guarantees.

\subsection{Related work}
\label{sec:related_work}
\looseness=-1 When the plant is uncertain, robust and gain-scheduled
synthesis certifies a controller against a prescribed uncertainty
set~\cite{zhou1996robust,shamma1988gain}, at the price of a new design per set. Learning-based treatments search for the LQR gain instead. Policy gradient provably reaches the optimal gain~\cite{fazel2018global}. The search has been extended to collections of plants through meta-learned initializations~\cite{toso2024meta}, federated updates~\cite{wang2023federated}, a common gain with sub-optimality and stability bounds~\cite{stamouli2025policy}, and imitation across related linear systems~\cite{zhang2023multi}. The guarantees in this line
attach to \emph{gains}: what is learned is a per-task matrix, a warm start,
or a single compromise gain whose optimality gap scales with task
heterogeneity, and the state dimension is fixed throughout. The object we
certify is different---one nonlinear sequence policy serving plants of
several dimensions and objectives.

A second line of work connects imitation error to closed-loop behavior. For incrementally stable experts, a small loss on expert-generated trajectories can bound the imitation loss along trajectories generated by the learned policy~\cite{pfrommer2022tasil,tu2022sample}. These single-system results apply to contracting experts and do not directly certify LQR cost. More closely related to our approach, \cite{hertneck2018learning} statistically validates a learned controller for approximate MPC using held-out rollouts, in the spirit of the scenario approach~\cite{campi2008exact}. Whereas that work certifies a probabilistic pointwise property, we certify closed-loop LQR performance over a finite validation horizon. An exact LQR identity yields a per-rollout excess-cost statistic, and independent validation rollouts provide a binomial bound on its threshold-violation probability. We complement this statistical certificate with Lyapunov-based regional guarantees for closed-loop containment.

On the architecture side, transformers can represent broad classes of sequence-to-sequence maps~\cite{yun2019are}. Theoretical guarantees are beginning to emerge for their use in estimation. For example, \cite{du2023can} shows that a transformer trained on random linear systems can recover the optimal filter in context and bounds its excess prediction risk. These results do not address closed-loop control. In contrast, we certify the cost of the closed loop generated by the learned policy, which a prediction-loss bound alone does not control.

\subsection{Contributions}
\looseness=-1 This letter builds on~\cite{binmohaya2026transformers}, which introduced and evaluated the transformer policy through simulations, but did not provide formal closed-loop guarantees. The present work extends~\cite{binmohaya2026transformers} with three theoretical contributions. First, we derive a finite-sample excess-risk bound for the masked Cauchy loss minimized during training (Theorem~\ref{thm:gen}). The result adapts the covering-number analysis of~\cite{du2023can} to cross-system control imitation. It bounds the imitation loss on expert-generated trajectories but does not by itself certify closed-loop behavior. Second, for each fixed problem instance, we derive regional closed-loop guarantees under stated operating-region assumptions. These consist of a forward-invariant region (Lemma~\ref{lem:roa}) and a deterministic worst-case bound on deviation from the optimal rollout (Theorem~\ref{thm:tube}). Third, and most importantly, we establish a probabilistic certificate for finite-horizon closed-loop near-optimality (Theorem~\ref{thm:nearopt}). Using independent calibration and validation rollouts, the certificate provides a high-confidence bound on the probability that the relative excess cost exceeds a calibrated threshold. It is computable from rollout data and the problem data $(A,B,Q,R)$. We evaluate this certificate on all $28$ benchmark systems introduced in~\cite{binmohaya2026transformers}.

\section{Problem Formulation}
\label{sec:problem_formulation}

Consider the LTI system
\small \begin{align}
    x_{t+1} = A x_{t} + B u_{t},
    \label{eq:system}
\end{align} \normalsize
with the fully measurable state $x_t \in \R^{n_x}$ and input $u_t \in \R^{n_u}$. In addition, define the infinite-horizon quadratic cost
\small \begin{align}
    \mathcal{J}(x_0;\pi) := \sum_{t=0}^{\infty} \bigl( x_t^\top Q x_t + u_t^\top R u_t \bigr),
    \label{eq:cost}
\end{align} \normalsize
where $Q\succ0$ and $R\succ0$ are diagonal with positive entries. The cost is evaluated under a causal policy $\pi$ generating $u_t$ from states observed up to time $t$. We further assume $(A,B)$ stabilizable and note that detectability of $(Q^{1/2},A)$ follows from $Q\succ0$. These are the standard assumptions for the discrete-time LQR problem~\cite{anderson1990optimal}. The solution is $u_t^\star=\pi^\star(x_t)=-K^\star x_t$ with
\small \begin{align}
    \label{eq:lqr_gain}
    K^\star = \bigl(R + B^\top P B\bigr)^{-1} B^\top P A.
\end{align} \normalsize
Here $P$ solves the Algebraic Riccati Equation (ARE)
\small \begin{align}
    P = Q + A^\top P A - A^\top P B \bigl(R + B^\top P B\bigr)^{-1} B^\top P A.
\label{ARE}
\end{align} \normalsize
The optimal cost is $\mathcal{J}(x_0;\pi^\star) = x_0^\top P x_0$.

We consider a family of $N$ such problems $\{A^{(i)},B^{(i)},Q^{(i)},R^{(i)}\}_{i=1}^N$ with possibly different state and input dimensions $(n_x^{(i)},n_u^{(i)})$. We assume each is well-posed with optimal gain $K^{\star,(i)}$ and controller $\pi^{\star,(i)}$. The superscript denotes a problem instance and is suppressed when a single instance is fixed.

\looseness=-1 The object of study is a \emph{single} learned policy $\pi_{\hat\phi}\bigl(\{x_l\}_{l=t-w}^{t},Q,R\bigr)$, parameterized by $\hat\phi\in\R^{n_\phi}$, which maps a window of $w{+}1$ recent states ($w\in\mathbb{N}$ the window length) and the instance's cost matrices to a control action, using no plant matrices at inference. The parameters $\hat\phi$ are obtained in our initial work~\cite{binmohaya2026transformers} by supervised imitation of optimal rollouts across the family, as a tractable surrogate for the ideal objective
\small \begin{align}
    \argmin_{\phi}\ \sum_{i=1}^N \sum_{j=1}^J \left[ \mathcal{J}^{(i)}\big(x_0^{(i,j)};\pi_\phi\big) - \mathcal{J}^{(i)}\big(x_0^{(i,j)};\pi^{\star,(i)}\big) \right],
    \label{eq:primary_objective}
\end{align} \normalsize
over $J$ initial conditions per instance. Here $\mathcal{J}^{(i)}$ is the LQR cost of instance $i$. Systems outside the training family receive fine-tuned copies of $\hat\phi$. Section~\ref{sec:approach} summarizes this pipeline.

\looseness=-1 Given the trained $\pi_{\hat\phi}$ and its fine-tuned copies, this letter asks three questions: \emph{(Q1) generalization}---does the training loss transfer from sampled rollouts to the population, with a finite-sample guarantee (Section~\ref{sec:theory}-A)? \emph{(Q2) closed-loop containment}---does the loop admit a forward-invariant operating region, and how far can it deviate from the optimal rollout (Section~\ref{sec:theory}-B)? \emph{(Q3) near-optimality}---can finitely many validation rollouts certify, with quantified confidence, that the closed-loop cost of $\pi_{\hat\phi}$, over a fixed horizon, stays close to the true optimal cost $\mathcal{J}(x_0;\pi^\star)$ (Sections~\ref{sec:theory}-C and \ref{sec:results})?

\section{The Learned Policy}
\label{sec:approach}
\looseness=-1 This section summarizes the elements of the framework in~\cite{binmohaya2026transformers} needed for the analysis: the shared data representation, the training loss, and the mapping from standardized policy outputs to physical controls. The transformer architecture enters the analysis only through the boundedness and parameter-Lipschitz properties in Lemma~\ref{lem:lip}. Architectural details are provided in~\cite{binmohaya2026transformers,vaswani2017attention}.

\subsection{Shared Representation}
\label{subsec:datacoll}\label{subsec:generalizable_optimal_control}
\looseness=-1 For each instance $i$, simulating the optimal feedback $\pi^{\star,(i)}(x_t)=-K^{\star,(i)}x_t$ from $J$ random initial conditions yields length-$T$ ($T\in\mathbb{N}$) state and control rollouts $\{x_t^{(i,j)}\}$ and $\{u_t^{(i,j)}\}$, $j=1,\ldots,J$; the instances span perturbed variants of mechanical, electrical, and robotic systems, each paired with several diagonal cost settings. Fixing an instance and dropping its superscript, the data take the following shared form.

\textbf{Standardization.} Let $\mu_x,\sigma_x\in\R^{n_x}$ and $\mu_u,\sigma_u\in\R^{n_u}$ be the entrywise means and standard deviations of the instance's states and inputs over all its rollouts ($\sigma_x,\sigma_u$ entrywise positive). Each sample is standardized entrywise, $\tilde{x}_t^{(j)}=(x_t^{(j)}-\mu_x)\oslash\sigma_x$ and $\tilde{u}_t^{(j)}=(u_t^{(j)}-\mu_u)\oslash\sigma_u$, where $\oslash$ denotes entrywise division.

\textbf{Padding and cost encoding.} Define $n_x^{\mathrm{max}}:=\max_i n_x^{(i)}$, $n_u^{\mathrm{max}}:=\max_i n_u^{(i)}$, and $\mathrm{Pad}(\cdot,\,n)$ as the operator appending trailing zeros up to dimension $n$. Next, define $\bar{x}_t^{(j)}=\mathrm{Pad}\big(\tilde{x}_t^{(j)},n_x^{\mathrm{max}}\big)$ and $\bar{u}_t^{(j)}=\mathrm{Pad}\big(\tilde{u}_t^{(j)},n_u^{\mathrm{max}}\big)$. The diagonal costs $Q=\mathrm{diag}(q)$, $R=\mathrm{diag}(r)$ are encoded by $\bar{q}:=\mathrm{Pad}(\log q,\,n_x^{\mathrm{max}})$ and $\bar{r}:=\mathrm{Pad}(\log r,\,n_u^{\mathrm{max}})$, with the logarithm entrywise. A sliding window of $w{+}1$ padded states is attached to the cost encoding to form the input
\small \begin{align}
S_t^{(j)} = \begin{bmatrix} (\bar{x}_{t-w}^{(j)})^\top \oplus {\bar{q}^\top} \oplus {\bar{r}^\top} \\ \vdots\\ (\bar{x}_{t}^{(j)})^\top \oplus {\bar{q}^\top} \oplus {\bar{r}^\top} \end{bmatrix} \label{eq:S}
\end{align} \normalsize
of fixed size $(w+1)\times d_{\mathrm{in}}$, $d_{\mathrm{in}}=2n_x^{\mathrm{max}}+n_u^{\mathrm{max}}$, where $\oplus$ denotes vector concatenation. The target is the current control $\bar{u}_t^{(j)}$.

\textbf{Masking.} The binary mask $\kappa:=\mathrm{Pad}(1_{n_u},\,n_u^{\mathrm{max}})$, where $1_n\in\R^n$ is the all-ones vector, marks the valid control entries. The dataset $\mathcal{D}$ collects the triplets $(S_t^{(i,j)},\bar{u}_t^{(i,j)},\kappa^{(i)})$ (instance superscripts restored) over all\footnote{For $t<w$ the missing past states are zero-padded to keep the length $w+1$.} $t\in\{0,\dots,T-1\}$, $j\in\{1,\dots,J\}$, and $i\in\{1,\dots,N\}$.

\subsection{Policy, Training, and Fine-Tuning}
\label{subsec:transformers}\label{subsec:training}\label{subsec:finetuning}
A transformer with parameters $\phi$ realizes the standardized policy $\bar{\pi}_\phi:\R^{(w+1)\times d_{\mathrm{in}}}\to\R^{n_u^{\mathrm{max}}}$: a finite composition of affine embeddings, multi-head self-attention blocks, layer normalizations, position-wise feed-forward maps, and a linear readout over parallel blocks that share the embedded input~\cite{binmohaya2026transformers,vaswani2017attention}. Training minimizes the \emph{masked Cauchy} loss over random mini-batches $\mathcal{B}$ of sample indices $(i,j,t)$ from $\mathcal{D}$,
\small \begin{align}
\mathcal{L}
= \frac{1}{|\mathcal{B}|} \sum_{(i,j,t) \in \mathcal{B}}
\ln \left( 1 + \left\| \frac{\kappa^{(i)}\odot(
\bar{\pi}_\phi(S_t^{(i,j)}) - \bar{u}_t^{(i,j)})}{\xi}\right\|_{2}^2 \right),
\nonumber
\end{align} \normalsize
\looseness=-1 with scale $\xi>0$ and $\odot$ the entrywise product, so that only valid control dimensions are penalized. At inference on an instance (superscript suppressed), the physical control is recovered by de-standardization of the first $n_u$ entries: $u_t=\sigma_u\odot\bar{\pi}_{\hat\phi}(S_t)(1{:}n_u)+\mu_u$. For systems outside the training family, a copy of the trained $\hat\phi$ (the \emph{base policy}) is fine-tuned on a few of their trajectories, one copy per system.

\section{Theoretical Analysis}
\label{sec:theory}
\looseness=-1 We give three guarantees: the excess-risk bound (Theorem~\ref{thm:gen}), the regional closed-loop properties (Lemma~\ref{lem:roa}, Theorem~\ref{thm:tube}), and, mainly, the near-optimality certificate over a validation horizon (Theorem~\ref{thm:nearopt}). The three results are complementary, and none relies on the conclusions of the others: the first concerns imitation loss on expert trajectories, the second gives deterministic regional guarantees for a fixed instance, and the third provides a statistical finite-horizon performance certificate over a specified rollout distribution.

\subsection{A generalization bound for the training loss}
\looseness=-1 We first bound the population training loss associated with the expert-rollout distribution. The unit of observation is an entire rollout: time-indexed samples within a rollout may be dependent, while rollouts are sampled independently. This result concerns training performance on expert-generated trajectories.

For instance $i\in\{1,\dots,N\}$, rollout $j\in\{1,\dots,J\}$, and time $t\in\{0,\dots,T-1\}$, let $S_t^{(i,j)}\in\R^{(w+1)\times d_{\mathrm{in}}}$, $\bar{u}_t^{(i,j)}\in\R^{n_u^{\mathrm{max}}}$, and $\kappa^{(i)}\in\{0,1\}^{n_u^{\mathrm{max}}}$ be the transformer input, the standardized padded LQR target, and the control mask of Section~\ref{subsec:generalizable_optimal_control}, and let $\bar{\pi}_\phi:\R^{(w+1)\times d_{\mathrm{in}}}\to\R^{n_u^{\mathrm{max}}}$, $\phi\in\Phi$, be the policy{, where $\Phi\subset\R^{n_\phi}$ is the parameter hypothesis class}. The per-sample loss is the summand of the training objective,
\small \begin{align}
{\mathcal{L}_t^{(i,j)}}(\phi) &:= \ln\!\Big(1+\tfrac{1}{\xi^2}\bigl\|\kappa^{(i)}\odot\bigl(\bar{\pi}_\phi(S_t^{(i,j)})-\bar{u}_t^{(i,j)}\bigr)\bigr\|_2^2\Big), \nonumber
\end{align} \normalsize
and the rollout-average, empirical, and population risks are $\bar{\mathcal{L}}^{(i,j)}(\phi) := \tfrac{1}{T}\!\sum_{t=0}^{T-1}\! {\mathcal{L}_t^{(i,j)}}(\phi)$,
$\widehat{\mathcal{L}}(\phi) := \tfrac{1}{NJ}\!\sum_{i,j}\!\bar{\mathcal{L}}^{(i,j)}(\phi)$, and $\mathcal{L}_{\mathrm{pop}}(\phi) := \tfrac{1}{N}\sum_{i=1}^{N}\mathbb{E}_{x_0\sim\nu_i}\bigl[\bar{\mathcal{L}}^{(i)}(\phi)\bigr]$, with $\xi>0$ the loss scale. Here, the expectation is over a new initial condition drawn from a distribution $\nu_i$ that is associated with instance $i$. Thus, $\mathcal{L}_{\mathrm{pop}}$ is the average expected imitation loss on new expert rollouts of the $N$ fixed benchmark instances. {We call the set of input--target pairs $(S_t^{(i,j)},\bar{u}_t^{(i,j)})$ realizable by such (zero-padded) optimal rollouts the \emph{rollout domain}.}

\begin{assumption}[Independent sampling]
\label{as:iid}
\looseness=-1 For each instance $i$, the initial conditions $x_0^{(i,1)},\ldots,x_0^{(i,J)}$ are drawn independently from $\nu_i$, and the resulting rollouts are independent across all pairs $(i,j)$. The rollout distributions may differ across instances. Time-indexed samples within a rollout need not be independent.
\end{assumption}
\begin{assumption}[Bounded data]
\label{as:bdd}
There is $B_S>0$ with $\|S_t^{(i,j)}\|_F\le B_S$ {($\|\cdot\|_F$ the Frobenius norm)} and $\|\bar{u}_t^{(i,j)}\|_2\le B_S$ for every rollout-domain sample.
\end{assumption}
\begin{assumption}[Bounded parameters]
\label{as:param}
There is $B_\Phi>0$ with $\|\phi\|_2\le B_\Phi$ for all $\phi\in\Phi$.
\end{assumption}

\begin{lemma}[Boundedness and parameter Lipschitzness]
\label{lem:lip}
Under Assumptions~\ref{as:bdd}--\ref{as:param}, there exist finite $B_\pi,{C_\Phi}>0$, depending only on $B_S$, $B_\Phi$, and the architecture, such that for every {input $S$ with $\|S\|_F\le B_S$} and all $\phi,\phi'\in\Phi$,
\small \begin{align}
\|\bar{\pi}_\phi(S)\|_2\le B_\pi,\quad
\|\bar{\pi}_\phi(S)-\bar{\pi}_{\phi'}(S)\|_2\le {C_\Phi}\|\phi-\phi'\|_2.
\nonumber
\end{align} \normalsize
\end{lemma}

\begin{proof}
\looseness=-1 On the domain of Assumptions~\ref{as:bdd}--\ref{as:param}, every block is bounded and Lipschitz in its input and its parameters (layer normalization has a fixed positive denominator, hence is smooth on bounded sets); composing these bounds gives $B_\pi$ and ${C_\Phi}$, as in~\cite{du2023can}.
\end{proof}

Lemma~\ref{lem:lip} induces the pseudometric $d_\Phi(\phi,\phi'):={C_\Phi}\|\phi-\phi'\|_2$, under which $\|\bar{\pi}_\phi(S)-\bar{\pi}_{\phi'}(S)\|_2\le d_\Phi(\phi,\phi')$ uniformly over rollout-domain inputs. Let $M(\Phi,d_\Phi,\varepsilon)$ denote the $\varepsilon$-covering number of $\Phi$ in $d_\Phi$, i.e., the size of the smallest $\varepsilon$-net~\cite{wainwright2019high}. For $\zeta\in(0,1)$ and $\varepsilon>0$ define
\small \begin{align}
\Lambda := \ln\!\Big(1+\tfrac{(B_\pi+B_S)^2}{\xi^2}\Big), \
{g_{\zeta,\varepsilon}} := \tfrac{2\varepsilon}{\xi}+\Lambda\sqrt{\tfrac{\ln\!\left(2M(\Phi,d_\Phi,\varepsilon)/\zeta\right)}{2NJ}}.
\nonumber
\end{align} \normalsize

\begin{theorem}[Excess-risk bound for the training loss]
\label{thm:gen}
Suppose Assumptions~\ref{as:iid}--\ref{as:param} hold and let $\hat{\phi}\in\argmin_{\phi\in\Phi}\widehat{\mathcal{L}}(\phi)$. Then for every $\zeta\in(0,1)$ and $\varepsilon>0$, with probability at least $1-\zeta$ over the sampled training rollouts,
\small \begin{align}
\mathcal{L}_{\mathrm{pop}}(\hat{\phi})-\inf_{\phi\in\Phi}\mathcal{L}_{\mathrm{pop}}(\phi)\ \le\ 2{g_{\zeta,\varepsilon}}.
\nonumber
\end{align} \normalsize
\end{theorem}

\begin{proof}[Proof sketch]
\looseness=-1 Each rollout loss lies in $[0,\Lambda]$ by Lemma~\ref{lem:lip} and Assumption~\ref{as:bdd}, so Hoeffding's inequality~\cite{hoeffding1963probability} with a union bound over an $\varepsilon$-net of $(\Phi,d_\Phi)$ bounds the uniform deviation by ${g_{\zeta,\varepsilon}}$; the $(1/\xi)$-Lipschitzness of the loss covers the net approximation~\cite{wainwright2019high}, and the risk decomposition at the empirical minimizer doubles the bound.
\end{proof}

Theorem~\ref{thm:gen} instantiates, for our masked control loss, the covering-number excess-risk analysis used for transformer filters in~\cite{du2023can}. \looseness=-1 The probability is taken over the sampled expert rollouts, with the $N$ problem instances held fixed. The theorem therefore controls excess risk on new rollouts of these instances. It implies neither closed-loop stability nor near-optimality, because in closed loop the states are generated by $\pi_{\hat\phi}$ rather than $\pi^\star$.

\subsection{Regional closed-loop properties}
\looseness=-1 We next derive deterministic regional guarantees for a fixed problem instance. Fix a single problem instance and suppress its index. $A,B,Q,R,K^\star$ {are as in Section~\ref{sec:problem_formulation}, and the control statistics $\mu_u,\sigma_u$ as in Sections~\ref{subsec:generalizable_optimal_control} and \ref{subsec:training}}. The optimal closed loop $A_{\mathrm{cl}}:=A-BK^\star$ is Schur, so there is a Lyapunov matrix $P_0\succ0$ solving $A_{\mathrm{cl}}^\top P_0A_{\mathrm{cl}}-P_0=-I$. Define $V_0(x):=x^\top P_0x$. With applied control {$u_t=\sigma_u\odot\bar{\pi}_{\hat\phi}(S_t)(1{:}n_u)+\mu_u$} and control deviation $e_t:=u_t+K^\star x_t$ from the optimal LQR action, de-standardization against the target {$\bar{u}_t=\mathrm{Pad}\bigl((-K^\star x_t-\mu_u)\oslash\sigma_u,\,n_u^{\mathrm{max}}\bigr)$} gives the {exact entrywise relation
$e_t=\sigma_u\odot\tilde{e}_t(1{:}n_u)$, $\tilde{e}_t:=\kappa\odot\bigl(\bar{\pi}_{\hat\phi}(S_t)-\bar{u}_t\bigr)$, whence $\|e_t\|_2\le\|\sigma_u\|_\infty\,\|\tilde{e}_t\|_2$, with $\|\sigma_u\|_\infty$ the largest entry of $\sigma_u$}.
By Lemma~\ref{lem:lip} and Assumption~\ref{as:bdd}, $\|\tilde{e}_t\|_2\le B_\pi+B_S$ whenever {the window's standardized input and target satisfy the bounds of Assumption~\ref{as:bdd}}, and then $\|e_t\|_2\le {c_e}:= {\|\sigma_u\|_\infty}(B_\pi+B_S)$. The bound is regional: the standardized target grows with $\|x_t\|_2$, so no constant bounds $\|e_t\|_2$ globally.

\begin{assumption}[Operating region]
\label{as:cl}
\looseness=-1 There is a radius $r_0>0$ such that every state window whose states all lie in $\mathcal X_0:=\{x:\|x\|_2\le r_0\}$ produces a standardized input and target obeying Assumption~\ref{as:bdd}, so that $\|e_t\|_2\le {c_e}={\|\sigma_u\|_\infty}(B_\pi+B_S)$ {for every such window} (for $t<w$ the missing slots are zero-padded, Section~\ref{subsec:generalizable_optimal_control}, and satisfy the bounds trivially). Outside $\mathcal X_0$, no uniform error bound is assumed.
\end{assumption}

\begin{lemma}[Forward-invariant operating region]
\label{lem:roa}
{Let Assumption~\ref{as:cl} hold. The learned closed loop obeys the perturbed dynamics $x_{t+1}=A_{\mathrm{cl}}x_t+Be_t$, with $\|e_t\|_2\le c_e$ whenever the current state window lies entirely in $\mathcal X_0$.} Let $\varrho:=\bigl(1-\tfrac{1}{\lambda_{\max}(P_0)}\bigr)^{1/2}\in[0,1)$ {(well defined: the Lyapunov equation gives $P_0\succeq I$, hence $\lambda_{\max}(P_0)\ge1$)} and $a_2:=\|B^\top P_0B\|_2$. While {the window remains in $\mathcal X_0$}, $V_0$ obeys the input-to-state-stability inequality~\cite{jiang2001input}, $V_0(x_{t+1})\le\varrho\,V_0(x_t)+\tfrac{a_2{c_e}^2}{1-\varrho}$,
with steady-state level $\eta_{\mathrm{ss}}:=a_2{c_e}^2/(1-\varrho)^2$. Whenever $\eta_{\mathrm{ss}}\le\lambda_{\min}(P_0)r_0^2$, any level $\eta$ with $\eta_{\mathrm{ss}}\le\eta\le\lambda_{\min}(P_0)\,r_0^2$ yields a sublevel set $\mathcal O_\eta:=\{x:V_0(x)\le\eta\}\subseteq\mathcal X_0$ that is invariant along trajectories started in it: every learned closed-loop trajectory with $x_0\in\mathcal O_\eta$ (and zero-padded initial window) satisfies $x_t\in\mathcal O_\eta$ for all $t\ge0$, as does the optimal rollout from the same $x_0$.
\end{lemma}

\begin{proof}
Write $\|z\|_{P_0}:=V_0(z)^{1/2}$ for $z\in\R^{n_x}$. For $x\in\mathcal O_\eta$, $V_0(x)\ge\lambda_{\min}(P_0)\|x\|_2^2$ and $\eta\le\lambda_{\min}(P_0)r_0^2$ give $\|x\|_2\le r_0$, so $\mathcal O_\eta\subseteq\mathcal X_0$, and $\|e_t\|_2\le {c_e}$ whenever the whole time-$t$ window lies in $\mathcal O_\eta$. The Lyapunov equation gives $V_0(A_{\mathrm{cl}}z)=V_0(z)-\|z\|_2^2\le\varrho^2V_0(z)$ (using $\|z\|_2^2\ge V_0(z)/\lambda_{\max}(P_0)$), i.e.\ $\|A_{\mathrm{cl}}z\|_{P_0}\le\varrho\|z\|_{P_0}$; and $\|Be_t\|_{P_0}^2\le a_2\|e_t\|_2^2\le a_2{c_e}^2$. The triangle inequality on $x_{t+1}=A_{\mathrm{cl}}x_t+Be_t$ gives $\|x_{t+1}\|_{P_0}\le\varrho\|x_t\|_{P_0}+\sqrt{a_2}{c_e}$; squaring and applying Young's inequality yields the ISS inequality (immediate when $\varrho=0$), whose scalar comparison recursion $V\mapsto\varrho V+a_2{c_e}^2/(1-\varrho)$ has fixed point $\eta_{\mathrm{ss}}$. Hence if $\eta\ge\eta_{\mathrm{ss}}$ and the whole time-$t$ window lies in $\mathcal O_\eta$, then $V_0(x_{t+1})\le\varrho\eta+(1-\varrho)\eta_{\mathrm{ss}}\le\eta$; induction from $t=0$, where the zero-padded slots satisfy the bounds trivially, keeps every state, hence every window, in $\mathcal O_\eta\subseteq\mathcal X_0$, proving the claim for the learned loop. For the optimal rollout $x_t^\star$ ($x_{t+1}^\star=A_{\mathrm{cl}}x_t^\star$, i.e., $e_t\equiv0$), $V_0(x_{t+1}^\star)\le\varrho^2V_0(x_t^\star)$ and $\mathcal O_\eta$ is also invariant.
\end{proof}

\begin{theorem}[Worst-case deviation from the optimal rollout]
\label{thm:tube}
{Let Assumption~\ref{as:cl} hold, suppose $\eta_{\mathrm{ss}}\le\lambda_{\min}(P_0)r_0^2$, and fix a level $\eta$ with $\eta_{\mathrm{ss}}\le\eta\le\lambda_{\min}(P_0)r_0^2$ and an $x_0\in\mathcal O_\eta$.} By Lemma~\ref{lem:roa}, the learned trajectory $x_t$ and the optimal rollout $x_t^\star$ ($x_{t+1}^\star=A_{\mathrm{cl}}x_t^\star$, $x_0^\star=x_0$) both remain in $\mathcal O_\eta$, so $\|e_t\|_2\le {c_e},\forall t$. With $a_1:=\|A_{\mathrm{cl}}^\top P_0B\|_2$,
\small \begin{align}
\|x_t-x_t^\star\|_2\ \le\ \chi:={c_e}\sqrt{\tfrac{2\,\lambda_{\max}(P_0)\,(2a_1^2+a_2)}{\lambda_{\min}(P_0)}},\quad \forall t.
\nonumber
\end{align} \normalsize
\end{theorem}

\begin{proof}
The deviation $z_t:=x_t-x_t^\star$ obeys $z_{t+1}=A_{\mathrm{cl}}z_t+Be_t$, $z_0=0$. With $\Omega_t:=z_t^\top P_0z_t$ and $A_{\mathrm{cl}}^\top P_0A_{\mathrm{cl}}=P_0-I$, expanding one step and bounding the cross and quadratic terms by Cauchy--Schwarz with $\|e_t\|_2\le {c_e}$ gives $\Omega_{t+1}-\Omega_t\le-\|z_t\|_2^2+2a_1{c_e}\|z_t\|_2+a_2{c_e}^2$. Split $-\|z_t\|_2^2$ in half: one half gives $-\Omega_t/(2\lambda_{\max}(P_0))$ via $\Omega_t\le\lambda_{\max}(P_0)\|z_t\|_2^2$, and the other bounds the forcing, $\max_{{s}\ge0}(-\tfrac12 {s}^2+2a_1{c_e s}+a_2{c_e}^2)=(2a_1^2+a_2){c_e}^2$. Hence $\Omega_{t+1}\le(1-{\gamma_0})\Omega_t+c_z$ with ${\gamma_0}:=1/(2\lambda_{\max}(P_0))\in(0,\tfrac12]$ and $c_z:=(2a_1^2+a_2){c_e}^2$; here $1-{\gamma_0}\in[\tfrac12,1)$ since $\lambda_{\max}(P_0)\ge1$ (Lemma~\ref{lem:roa}). Iterating from $\Omega_0=0$, $\Omega_t\le c_z/{\gamma_0}=2\lambda_{\max}(P_0)(2a_1^2+a_2){c_e}^2$, so $\|z_t\|_2^2\le\Omega_t/\lambda_{\min}(P_0)\le\chi^2$.
\end{proof}

\subsection{Probabilistic closed-loop near-optimality}
\label{subsec:nearopt}

\looseness=-1 We now develop a statistical certificate for closed-loop LQR performance over a fixed validation horizon. This result relies on neither the imitation-risk bound in Section~\ref{sec:theory}-A nor the deterministic regional guarantees in Section~\ref{sec:theory}-B. It requires neither a uniform action-error bound nor a verified invariant region. Instead, it combines an exact LQR cost identity with independent calibration and validation rollouts. We first fix one problem instance and suppress its index. Let $\mathcal{W}:=R+B^\top PB\succ0$, and denote the truncated cost over a horizon $\tau\in\mathbb{N}$ by $\mathcal{J}_\tau(x_0;\pi):=\sum_{t=0}^{\tau-1}\bigl(x_t^\top Qx_t+u_t^\top Ru_t\bigr)$, written $\mathcal{J}_\tau(x_0;u)$ when an explicit control sequence $\{u_t\}$ is applied in place of a policy. Since $x_0^\top Px_0=\mathcal{J}(x_0;\pi^\star)\ge x_0^\top Qx_0$ and $Q\succ0$, we have $P\succ0$.

\begin{lemma}[Exact cost decomposition]
\label{lem:cs}
For any control sequence $\{u_t\}$ applied to \eqref{eq:system}, with $e_t:=u_t+K^\star x_t$ and every $\tau\in\mathbb{N}$,
\begin{align}
\mathcal{J}_\tau(x_0;u) \;=\; x_0^\top Px_0 \;-\; x_\tau^\top Px_\tau \;+\; \sum_{t=0}^{\tau-1} e_t^\top \mathcal{W}\,e_t .
\label{eq:cs}
\end{align}
\end{lemma}

\begin{proof}
This is the classical completion of squares for the discrete-time LQR: expanding one step with the ARE \eqref{ARE} and telescoping gives \eqref{eq:cs}; see, e.g.,~\cite{anderson1990optimal}.
\end{proof}

\looseness=-1 Since $\mathcal{J}(x_0;\pi^\star)=x_0^\top Px_0$, identity \eqref{eq:cs} says the terminal-value-completed excess cost of \emph{any} policy is exactly the $\mathcal{W}$-weighted energy of its control deviations. Near-optimality therefore reduces to the error energy, which we certify directly from sampled rollouts.

For a closed-loop rollout of $\pi_{\hat\phi}$ from $x_0\neq0$ over a validation horizon $T_v\in\mathbb{N}$, define the excess statistic
\begin{align}
X(x_0)\;:=\;\frac{\sum_{t=0}^{T_v-1} e_t^\top\mathcal{W}\,e_t}{x_0^\top Px_0}.
\label{eq:excess}
\end{align}
The partial sums of the error energy are nondecreasing in the horizon, so Lemma~\ref{lem:cs} together with $x_\tau^\top Px_\tau\ge0$ gives
$\mathcal{J}_\tau(x_0;\pi_{\hat\phi})\le\bigl(1+X(x_0)\bigr)\,\mathcal{J}(x_0;\pi^\star)$ for every $\tau\le T_v$.
The whole rollout is a single sample; no independence across time is required. This design follows the statistical verification of~\cite{hertneck2018learning} and the scenario approach~\cite{campi2008exact}; the calibration-fixed threshold parallels split conformal prediction~\cite{vovk2005algorithmic}. Let $\nu_{\mathrm{cert}}$ be a distribution of initial conditions. A threshold $\bar{X}$ is fixed on an independent calibration draw, for which we take the largest statistic observed on $m_{\mathrm{cal}}$ calibration rollouts. For $x_0\sim\nu_{\mathrm{cert}}$, define the violation indicator $F(x_0):=1$ if $X(x_0)>\bar{X}$, and $F(x_0):=0$ otherwise.

\begin{lemma}[Clopper--Pearson certificate]
\label{lem:cert}
Draw $x_0^{(1)},\dots,x_0^{(m)}\overset{\mathrm{iid}}{\sim}\nu_{\mathrm{cert}}$ and let $k:=\sum_{l=1}^m F(x_0^{(l)})$ be the number of violating rollouts. For $\zeta_v\in(0,1)$ define the exact one-sided upper limit
\begin{align}
\alpha_m\;:=\;\inf\bigl\{p'\in[0,1]\,:\,\mathbb{P}\bigl(\mathrm{Bin}(m,p')\le k\bigr)\le\zeta_v\bigr\},
\label{eq:cert}
\end{align}
equivalently the $(1-\zeta_v)$-quantile of the $\mathrm{Beta}(k+1,\,m-k)$ distribution~\cite{clopper1934use}; for $k=0$, $\alpha_m=1-\zeta_v^{1/m}$, and $\alpha_m:=1$ for $k=m$. Then, with probability at least $1-\zeta_v$ over the validation draw, $\mathbb{P}_{x_0\sim\nu_{\mathrm{cert}}}\bigl(F(x_0)=1\bigr)\le\alpha_m$.
\end{lemma}

\begin{proof}
This is the classical one-sided Clopper--Pearson bound; the proof inverts the binomial tail~\cite{clopper1934use}.
\end{proof}

\begin{theorem}[Certified closed-loop near-optimality]
\label{thm:nearopt}
Fix $\bar{X}$ from the calibration draw, run the validation of Lemma~\ref{lem:cert}, and let $\alpha_m$ be as in \eqref{eq:cert}. Then, with probability at least $1-\zeta_v$ over the validation draw, a fresh initial condition $x_0\sim\nu_{\mathrm{cert}}$ with $x_0\neq0$ satisfies, with probability at least $1-\alpha_m$,
\begin{align}
\mathcal{J}_\tau(x_0;\pi_{\hat\phi})\;\le\;\bigl(1+\bar{X}\bigr)\,\mathcal{J}(x_0;\pi^\star)
\quad\text{for every }\tau\le T_v.
\label{eq:nearopt}
\end{align}
\end{theorem}

\begin{proof}
By Lemma~\ref{lem:cert}, with probability at least $1-\zeta_v$ over the validation draw, a fresh rollout satisfies $X(x_0)\le\bar{X}$ except with probability at most $\alpha_m$; condition on that event. For every $\tau\le T_v$, Lemma~\ref{lem:cs} with $x_\tau^\top Px_\tau\ge0$ gives
$\mathcal{J}_\tau(x_0;\pi_{\hat\phi})-x_0^\top Px_0\le\sum_{t=0}^{\tau-1}e_t^\top\mathcal{W}e_t\le\sum_{t=0}^{T_v-1}e_t^\top\mathcal{W}e_t=X(x_0)\,x_0^\top Px_0\le\bar{X}\,x_0^\top Px_0$,
and $x_0^\top Px_0=\mathcal{J}(x_0;\pi^\star)$.
\end{proof}

\begin{corollary}[Learned-then-optimal policies]
\label{cor:hybrid}
In the setting of Theorem~\ref{thm:nearopt}, with the same probabilities, for every $\tau\le T_v$ the hybrid policy that applies $\pi_{\hat\phi}$ up to time $\tau$ and $\pi^\star$ thereafter has infinite-horizon cost at most $\bigl(1+\bar{X}\bigr)\,\mathcal{J}(x_0;\pi^\star)$; this cost equals $\mathcal{J}_\tau(x_0;\pi_{\hat\phi})+x_\tau^\top Px_\tau$.
\end{corollary}
\begin{proof}
The hybrid cost equals $\mathcal{J}_\tau(x_0;\pi_{\hat\phi})+x_\tau^\top Px_\tau$ because $x_\tau^\top Px_\tau$ is the optimal cost-to-go; by \eqref{eq:cs} this is $x_0^\top Px_0+\sum_{t<\tau}e_t^\top\mathcal{W}e_t\le(1+\bar{X})\,x_0^\top Px_0$ on the certified event.
\end{proof}

\looseness=-1 For $\pi_{\hat\phi}$ itself the tail beyond $T_v$ is not certified; Corollary~\ref{cor:hybrid} closes the horizon by the switch to $\pi^\star$, whose remaining optimal cost-to-go is the terminal term $x_\tau^\top Px_\tau$ of \eqref{eq:cs}. Implementing this switch requires the model-based gain $K^\star$, hence plant knowledge at deployment. On the validation draws of Section~\ref{sec:results}, the median remaining optimal cost-to-go at $T_v$ is below $10^{-2}$ of $x_0^\top Px_0$ on $24$ of the $28$ systems; the four slowest systems retain more. Nothing in Lemma~\ref{lem:cert} or Theorem~\ref{thm:nearopt} uses the structure of $\nu_{\mathrm{cert}}$: when each draw is a tuple $(A,B,Q,R,x_0)$ from a joint certification distribution $\mathcal{D}_{\mathrm{cert}}$, with $X$ evaluated through the sampled instance's own $(P,K^\star,\mathcal{W})$, every rollout is still one i.i.d.\ Bernoulli draw and the certificate holds verbatim.

\section{Results}
\label{sec:results}

\looseness=-1 We instantiate Theorem~\ref{thm:nearopt} separately on each of the $28$ benchmark systems from~\cite{binmohaya2026transformers}, yielding $28$ per-system family certificates. Each rollout draws an independent tuple $(A,B,Q,R,x_0)\sim\mathcal{D}_{\mathrm{cert}}$. Before discretization, each physical parameter is perturbed according to $p_j=p_{j,\mathrm{nom}}(1+\delta_j)$, where $\delta_j\overset{\mathrm{iid}}{\sim}\mathrm{Unif}[-0.1,0.1]$. The diagonal cost pair is drawn uniformly from the system's nine training cost pairs, and $x_0$ is drawn from its training initial-condition distribution. Thus, the certified violation probability is joint over the plant, cost, and initial condition. For each sampled instance, $(P,K^\star,\mathcal{W})$ are computed from the discrete-time model realized by the simulator. The initial-condition distribution is the same as that used in training and fine-tuning, except that the Segway uses half of the default family region. This choice was fixed before calibration and validation.

For each system, the threshold $\bar{X}$ is set to the largest excess statistic observed among $m_{\mathrm{cal}}=200$ independent calibration rollouts. We then draw $m=1000$ independent validation instances, simulate each over $T_v=500$ steps, and count the violations of $X>\bar{X}$. At confidence level $1-\zeta_v=0.95$, the exact binomial bound gives $\alpha_m=0.003$ when no violation is observed. Across all systems, the per-system bounds satisfy $\alpha_m\le0.031$. Applying a Bonferroni correction with $\zeta_v/28$ provides $95\%$ simultaneous confidence for all $28$ certificates, with every corrected violation-probability bound below $0.038$.

\autoref{tab:cert} reports the per-system certificates. Nine systems certify a relative-excess threshold below $1\%$, twenty certify a threshold below $10\%$, and twenty-six have thresholds below one. The tightest threshold is $4.8\times10^{-6}$ for the suspension system. The three most challenging unseen systems use additional system-specific fine-tuning. This was completed before calibration and validation. Although these systems have the largest thresholds in \autoref{tab:cert}, their certified violation-probability bounds remain small.

\looseness=-1 \autoref{fig:invariant} illustrates a statistical containment certificate for the nominal instances of four two-state benchmark systems, three seen and one unseen. For each system, the shaded region is $\mathcal O_{\eta_0}$, where $\eta_0:=\max_{j}V_0(x_0^{(j)})$ is computed from $200$ independent calibration draws and fixed before validation. A failure occurs if $x_t\notin\mathcal O_{\eta_0}$ for some $t\le T_v$. An initial state outside the region is therefore a failure at $t=0$. Among $1000$ independent validation rollouts per system, every trajectory starting inside its calibrated region remains inside throughout the validation horizon. The observed failures, at most five per system, all result from initial states outside the region. By Lemma~\ref{lem:cert}, the corresponding containment-failure probability is at most $0.011$ with $95\%$ confidence for each system.

\looseness=-1 In each panel, $x_1$ and $x_2$ denote angle and angular velocity for the pendulum, position and velocity for the double integrator and damped oscillator, and prey and predator deviations from equilibrium for the linearized Lotka--Volterra model. The displayed learned trajectories remain close to their optimal counterparts and converge toward the origin within the calibrated region. This is a statistical containment certificate. This is distinct from the deterministic guarantee of Lemma~\ref{lem:roa}, whose constants are not evaluated here.

\begin{figure}[t]
\centering
\includegraphics[width=0.88\linewidth]{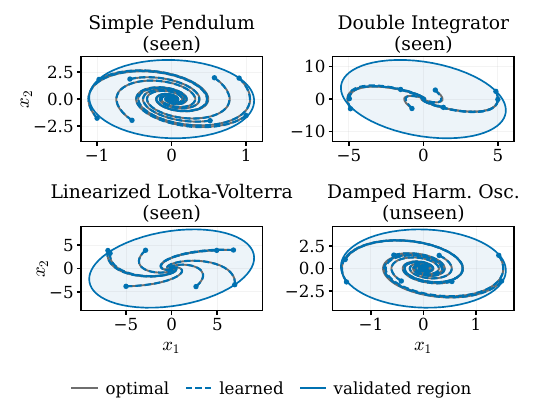}
\caption{Validated containment regions ($V_0(x)\le\eta_0$); eight of $1000$ validation rollouts per system; dots mark initial states.}
\label{fig:invariant}
\end{figure}

\begin{table}[t]
\scriptsize
\renewcommand{\arraystretch}{0.85}
\setlength{\tabcolsep}{3.5pt}
\centering
\caption{Family-Level Near-Optimality Certificates (Theorem~\ref{thm:nearopt})}
\label{tab:cert}
\begin{tabular}{@{}lrrrrr@{}}
\toprule
System & $\bar{X}$ & $\hat p$ & $\alpha_m$ & $X$ med & $X$ max\\
\midrule
1.\ Simple Pendulum (s) & $2.0{\times}10^{-2}$ & 0.002 & 0.006 & $6.9{\times}10^{-4}$ & $5.6{\times}10^{-2}$\\
2.\ Two Link Arm (s) & $3.6{\times}10^{-3}$ & 0.012 & 0.019 & $1.1{\times}10^{-3}$ & $4.5{\times}10^{-3}$\\
3.\ Spring Damper (s) & $3.3{\times}10^{-3}$ & 0.020 & 0.029 & $4.1{\times}10^{-4}$ & $1.0{\times}10^{-2}$\\
4.\ Suspension (s) & $4.8{\times}10^{-6}$ & 0.016 & 0.024 & $1.1{\times}10^{-6}$ & $5.4{\times}10^{-6}$\\
5.\ DC Motor (s) & $5.3{\times}10^{-2}$ & 0.001 & 0.005 & $4.2{\times}10^{-5}$ & $7.1{\times}10^{-2}$\\
6.\ Three Link Man.\ (s) & $4.1{\times}10^{-3}$ & 0.008 & 0.014 & $1.1{\times}10^{-3}$ & $6.6{\times}10^{-3}$\\
7.\ Diff.\ Drive (s) & $1.5{\times}10^{-2}$ & 0.002 & 0.006 & $4.2{\times}10^{-4}$ & $1.8{\times}10^{-2}$\\
8.\ SCARA (s) & $4.26{\times}10^{-2}$ & 0.001 & 0.005 & $2.0{\times}10^{-3}$ & $4.31{\times}10^{-2}$\\
9.\ Omnidirectional (s) & $6.2{\times}10^{-3}$ & 0.005 & 0.010 & $4.4{\times}10^{-4}$ & $7.7{\times}10^{-3}$\\
10.\ Cable Driven (s) & $4.7{\times}10^{-3}$ & 0.011 & 0.018 & $4.3{\times}10^{-4}$ & $6.7{\times}10^{-3}$\\
11.\ Flexible Joint (s) & $2.4{\times}10^{-1}$ & 0.003 & 0.008 & $1.3{\times}10^{-1}$ & $2.6{\times}10^{-1}$\\
12.\ Six DOF Man.\ (s) & $1.5{\times}10^{-3}$ & 0.008 & 0.014 & $4.0{\times}10^{-4}$ & $2.3{\times}10^{-3}$\\
13.\ Dual Arm (s) & $1.6{\times}10^{-2}$ & 0.003 & 0.008 & $1.4{\times}10^{-3}$ & $2.9{\times}10^{-2}$\\
14.\ Double Integr.\ (s) & $3.1{\times}10^{-2}$ & 0.002 & 0.006 & $6.2{\times}10^{-4}$ & $3.3{\times}10^{-2}$\\
15.\ Lotka Volterra (s) & $3.3{\times}10^{-1}$ & 0.000 & 0.003 & $7.6{\times}10^{-4}$ & $2.4{\times}10^{-1}$\\
16.\ Inverted Pend.\ (u) & $1.1{\times}10^{1}$ & 0.006 & 0.012 & $4.4{\times}10^{-2}$ & $6.2{\times}10^{1}$\\
17.\ Segway (u) & $6.0{\times}10^{0}$ & 0.001 & 0.005 & $1.3{\times}10^{-1}$ & $1.5{\times}10^{1}$\\
18.\ Asym.\ Oscillator (u) & $8.7{\times}10^{-1}$ & 0.000 & 0.003 & $2.4{\times}10^{-2}$ & $7.1{\times}10^{-1}$\\
19.\ Active Mass D.\ (u) & $1.1{\times}10^{-4}$ & 0.001 & 0.005 & $1.4{\times}10^{-5}$ & $1.2{\times}10^{-4}$\\
20.\ Coupled Osc.\ (u) & $1.1{\times}10^{-2}$ & 0.004 & 0.009 & $1.1{\times}10^{-3}$ & $1.4{\times}10^{-2}$\\
21.\ Damped Osc.\ (u) & $1.5{\times}10^{-2}$ & 0.003 & 0.008 & $1.3{\times}10^{-3}$ & $7.4{\times}10^{-2}$\\
22.\ Triple Mass Spr.\ (u) & $2.0{\times}10^{-2}$ & 0.021 & 0.030 & $3.3{\times}10^{-3}$ & $5.0{\times}10^{-2}$\\
23.\ Electromech.\ Act.\ (u) & $4.8{\times}10^{-1}$ & 0.001 & 0.005 & $3.0{\times}10^{-3}$ & $6.6{\times}10^{-1}$\\
24.\ Thermal (u) & $1.2{\times}10^{-3}$ & 0.011 & 0.018 & $7.5{\times}10^{-5}$ & $5.4{\times}10^{-3}$\\
25.\ Fluid Tank (u) & $9.6{\times}10^{-2}$ & 0.010 & 0.017 & $2.3{\times}10^{-3}$ & $1.3{\times}10^{0}$\\
26.\ Vibrating Beam (u) & $2.1{\times}10^{-2}$ & 0.018 & 0.027 & $9.7{\times}10^{-3}$ & $5.7{\times}10^{-2}$\\
27.\ Motor Generator (u) & $1.6{\times}10^{-1}$ & 0.000 & 0.003 & $8.5{\times}10^{-4}$ & $1.2{\times}10^{-1}$\\
28.\ Mech.\ Linkage (u) & $1.4{\times}10^{-1}$ & 0.006 & 0.012 & $1.4{\times}10^{-2}$ & $2.2{\times}10^{-1}$\\
\bottomrule
\end{tabular}
\par
\begin{minipage}{\linewidth}\scriptsize
$\bar{X}$: calibration-fixed relative excess threshold; $\hat p=k/m$: observed violation rate; $\alpha_m$: confidence upper bound on the violation probability; $X$ med/max: median and maximum of \eqref{eq:excess} on validation; (s) seen, base policy; (u) unseen, fine-tuned copy.
\end{minipage}
\end{table}

\section{Conclusions}
\label{sec:conclusions}

\looseness=-1 This letter develops three guarantees for the transformer policy introduced in~\cite{binmohaya2026transformers}. First, a finite-sample bound controls the population training loss. Second, regional closed-loop results establish a forward-invariant operating region and a worst-case bound on deviation from the optimal rollout. Third, and most importantly, an exact LQR identity expresses the closed-loop excess cost as a per-rollout statistic, while an exact binomial bound yields a high-confidence near-optimality certificate over the validation horizon. Across $28$ benchmark systems, twenty certify suboptimality below $10\%$, and all per-system certificates have violation probabilities below $3.1\%$, each at $95\%$ confidence. These guarantees apply to the validated distributions over plants, costs, and initial conditions. Future work will address infinite-horizon guarantees, broader certification distributions, pooled certificates across heterogeneous systems, and noisy, model-mismatched, partially observed, or non-diagonally weighted settings.

\bibliographystyle{IEEEtran}
\bibliography{lcsys}

\end{document}